\documentclass{article}
\usepackage{graphicx} 
\usepackage[authoryear]{natbib}
\usepackage{amsthm}   
\usepackage{amsmath}
\usepackage{amsmath}   
\usepackage{amsfonts}  

\numberwithin{equation}{section}
\numberwithin{figure}{section}

\newtheorem{theorem}{Theorem}[section]
\newtheorem{definition}[theorem]{Definition}
\newtheorem{example}[theorem]{Example}

\newcommand{\Corr}{\text{Corr}}
\newcommand{\Cov}{\text{Cov}}

\title{An Introduction to Copulas: a Complement}
\author{Werner G. Müller}
\date{May 2026}

\begin{document}

\maketitle

\begin{abstract}
For many years I have taught an advanced statistical inference course for master’s students using the text of \citet{CasellaBerger2002}. The book gives a comprehensive treatment of the core topics at a level that avoids measure theory while remaining mathematically precise, but it does not cover the increasingly important concept of copulas. The present notes are intended to complement the book by adding two sections on copulas in a style that is as close as possible to that of the original text. Numbering of definitions, theorems, examples, and exercises is consistent with \citet{CasellaBerger2002}, but the material may also be read as a brief, stand-alone introduction to copula theory.
\end{abstract}

 \setcounter{section}{9}
 \renewcommand{\thesection}{4.\arabic{section}}

\section{Copulas}
\label{sec:4.8}

In many applications we wish to model the joint behavior of several random
variables while allowing flexible and possibly different models for their
marginal distributions. Classical multivariate families (for example, the
bivariate normal in Section~3.3) tie the joint and marginal distributions
together quite rigidly. Copulas provide a way to separate these two aspects: the margins and
the dependence structure \citep{Nelsen1999,Joe1997}.

In the years leading up to the 2008 financial crisis, a particular dependence
model---the Gaussian copula---became widely used on Wall Street to price and
manage the risk of complex credit products such as collateralized debt
obligations \citep{Salmon2009}. By summarizing default dependence between
thousands of loans in a single, tractable parameter, it allowed rating agencies
and banks to turn heterogeneous pools of mortgages into highly rated
securities. The simplicity of the formula helped mask the
fact that it relied on a very specific and fragile assumption about joint tail
behavior: that extreme losses across different loans were, in effect, nearly
independent once the linear correlation was set. When housing
prices fell nationwide and defaults became simultaneously likely across many
regions, this assumption broke down, and the models severely understated the
probability of joint defaults. Copulas themselves are not at
fault, but this episode illustrates how important it is to understand what a
dependence model captures---and, just as importantly, what it leaves out.

In the following we consider solely the bivariate case; higher dimensions are similar in spirit but
technically more involved \citep{Nelsen1999,Joe1997}.

\subsection{Definition and Basic Properties}
\label{subsec:4.8.1}

Recall from Section~4.5 that a bivariate distribution function $F_{X,Y}$ must be
2-increasing and right-continuous and must have appropriate limits at
infinity, see Definition~4.1.1 and the discussion preceding
Example~4.5.1. Copulas are special bivariate distribution functions on
$(0,1)^2$ with uniform marginals \citep{Nelsen1999}.

\begin{definition}\label{def:4.8.1}
Let $I=(0,1)$. A \textbf{bivariate copula} (or \textbf{2-copula}) is a function
$C:I^2\to I$ such that:
\begin{enumerate}
\item[(i)] (Margins) For every $u,v\in I$,
\[
C(u,0)=0, C(0,v)=0\quad  (\text{i.e. } C \text{ is grounded}),\quad C(u,1)=u, C(1,v)=v.
\]
\item[(ii)] (2-increasing) For every $u_1\le u_2$ and $v_1\le v_2$ in $I$,
\[
C(u_2,v_2)-C(u_2,v_1)-C(u_1,v_2)+C(u_1,v_1)\ge 0.
\]
\end{enumerate}
\end{definition}

Thus a copula is itself a bivariate cdf on $(0,1)^2$ with uniform $(0,1)$
marginals, in the sense of Definition~4.1.1 \citep{Nelsen1999}. Three special copulas will be useful: 
\begin{definition}\label{def:4.8.2}
Let $I=(0,1)$. The following functions from $I^2$ to $I$ are called the
\emph{product}, \emph{lower Fr\'echet}, and \emph{upper Fr\'echet} copulas,
respectively:
\begin{align*}
\Pi(u,v) &= uv, \\
W(u,v)  &= \max\{u+v-1,0\}, \\
M(u,v)  &= \min\{u,v\},
\end{align*}
for $u,v\in I$.
\end{definition}

These correspond, respectively, to independence, perfect negative monotone
dependence, and perfect positive monotone dependence \citep{Nelsen1999}. As in the
inequalities for general bivariate distributions (Section~4.5), one can show
that for any copula $C$,
\[
W(u,v)\le C(u,v)\le M(u,v),\qquad u,v\in I.
\]
\noindent\textbf{Proof:}
Let $C$ be a copula on $I^2$, and fix $u,v\in I$. Since $C$ is a bivariate distribution function, it is nondecreasing in each coordinate. If $u\le v$, then by monotonicity in the second coordinate and the boundary condition $C(u,1)=u$,
\[
C(u,v)\le C(u,1)=u=\min\{u,v\}.
\]
If $v\le u$, then by monotonicity in the first coordinate and $C(1,v)=v$,
\[
C(u,v)\le C(1,v)=v=\min\{u,v\}.
\]
Thus $C(u,v)\le M(u,v)=\min\{u,v\}$. For the lower bound, use the 2-increasing property with $u_1=u$, $u_2=1$, $v_1=v$, and $v_2=1$ to obtain
\[
C(1,1)-C(1,v)-C(u,1)+C(u,v)\ge 0.
\]
Using the boundary conditions $C(1,1)=1$, $C(1,v)=v$, and $C(u,1)=u$, this becomes
\[
1-v-u+C(u,v)\ge 0,
\]
or $C(u,v)\ge u+v-1$. Because any copula is a cdf, $C(u,v)\ge 0$, so
\[
C(u,v)\ge\max\{u+v-1,0\}=W(u,v).
\]
Combining the two inequalities yields $W(u,v)\le C(u,v)\le M(u,v)$ for all $u,v\in I$. \qed

The copula $\Pi$ corresponds to the case of independent random variables (see
Definition~1.3.3); $M$ and $W$ correspond to the Fr\'echet bounds in the general
dependence ordering \citep{Nelsen1999}.

\subsection{Sklar's Theorem}
\label{subsec:4.8.2}

The central result connecting copulas with joint distributions is due to \citet{Sklar1959}
. It shows how any joint cdf $F_{X,Y}$
studied in Sections~4.1 and~4.5 can be decomposed into its marginals
(Sections~4.1 and~4.2) and a copula.

\begin{theorem}[Sklar's Theorem]\label{thm:4.8.2}
Let $F_{X,Y}$ be the joint cdf of a pair of random variables $(X,Y)$ with
marginal cdfs $F_X$ and $F_Y$.
\begin{enumerate}
\item[(a)] There exists a copula $C$ such that
\[
F_{X,Y}(x,y)=C\bigl(F_X(x),F_Y(y)\bigr),\qquad x,y\in\mathbb{R}.
\]
\item[(b)] If $F_X$ and $F_Y$ are continuous, then $C$ is unique.
\end{enumerate}
Conversely, if $C$ is any copula and $F_X,F_Y$ are univariate cdfs, then
\[
F(x,y)=C\bigl(F_X(x),F_Y(y)\bigr)
\]
is a bivariate cdf with marginals $F_X$ and $F_Y$.
\end{theorem}
\noindent\textbf{Proof:}
For $x,y\in\mathbb{R}$, define
\[
u=F_X(x),\qquad v=F_Y(y),
\]
and let $H(x,y)=F_{X,Y}(x,y)$ denote the joint cdf of $(X,Y)$. For part (a), we
must construct a function $C$ on $[0,1]^2$ such that $H(x,y)=C(F_X(x),F_Y(y))$.

First suppose that $F_X$ and $F_Y$ are continuous and strictly increasing. Then
each cdf has a (generalized) inverse $F_X^{-1}$ and $F_Y^{-1}$, and for any
$u,v\in(0,1)$ we may define
\[
C(u,v)=H\bigl(F_X^{-1}(u),F_Y^{-1}(v)\bigr).
\]
If $u=F_X(x)$ and $v=F_Y(y)$, then $F_X^{-1}(u)=x$ and $F_Y^{-1}(v)=y$, so
\[
C\bigl(F_X(x),F_Y(y)\bigr)
=H(x,y)
=F_{X,Y}(x,y),
\]
which gives the desired representation. It is routine to check from the
properties of $H$ and the marginals that $C$ satisfies the copula axioms and
has uniform margins on $(0,1)$.

For the general case, when $F_X$ and $F_Y$ need not be strictly increasing, we
use generalized inverses. For $0<u<1$, define
\[
F_X^{-1}(u)=\inf\{x:F_X(x)\ge u\},\qquad
F_Y^{-1}(v)=\inf\{y:F_Y(y)\ge v\},
\]
and set
\[
C(u,v)
=H\bigl(F_X^{-1}(u),F_Y^{-1}(v)\bigr),\qquad 0<u,v<1.
\]
Using the fact that $F_X(F_X^{-1}(u))\ge u$ and that $F_X^{-1}(F_X(x))\le x$
(and similarly for $F_Y$), together with the right-continuity and monotonicity
properties of $H$, one checks that
\[
H(x,y)
=C\bigl(F_X(x),F_Y(y)\bigr)
\]
for all $x,y$. Again, it follows from the construction that $C$ is grounded,
has uniform margins, and is 2-increasing, hence is a copula.

For part (b), suppose in addition that $F_X$ and $F_Y$ are continuous. If
$C_1$ and $C_2$ are copulas such that
\[
H(x,y)
=C_1\bigl(F_X(x),F_Y(y)\bigr)
=C_2\bigl(F_X(x),F_Y(y)\bigr)
\]
for all $x,y$, then for any $u,v\in(0,1)$ we can write $u=F_X(x)$ and
$v=F_Y(y)$ for some $x,y$, and hence
\[
C_1(u,v)=C_2(u,v).
\]
Thus $C_1=C_2$ on $(0,1)^2$, so the copula is unique when the marginals are
continuous.

For the converse, let $C$ be any copula and let $F_X$ and $F_Y$ be univariate
cdfs. Define
\[
F(x,y)=C\bigl(F_X(x),F_Y(y)\bigr).
\]
Since $C$ is a bivariate cdf on $[0,1]^2$ and $F_X,F_Y$ are cdfs on
$\mathbb{R}$, the composition $F$ is nondecreasing in each argument, right
continuous, and has the correct limits at $\pm\infty$, so $F$ is a bivariate
cdf. Moreover,
\[
\lim_{y\to\infty}F(x,y)
=\lim_{y\to\infty}C\bigl(F_X(x),F_Y(y)\bigr)
=C\bigl(F_X(x),1\bigr)
=F_X(x),
\]
and similarly $\lim_{x\to\infty}F(x,y)=F_Y(y)$, so the marginals of $F$ are
$F_X$ and $F_Y$. This completes the proof. \qed

\begin{example}[Exponential margins with dependence]\label{ex:4.8.1}
Let $X$ and $Y$ be nonnegative random variables with exponential marginals
\[
F_X(x)=1-e^{-\lambda x},\quad F_Y(y)=1-e^{-\mu y},\qquad x,y\ge 0,
\]
and suppose their joint cdf is
\[
F_{X,Y}(x,y)=C\bigl(F_X(x),F_Y(y)\bigr),
\]
for some copula $C$.

\begin{enumerate}
\item[(a)] If $C(u,v)=uv$ (the product copula), then
$F_{X,Y}(x,y)=F_X(x)F_Y(y)$ for all $x,y\ge 0$, and $X$ and $Y$ are independent
(Section~1.3).
\item[(b)] If instead we take, for some $\theta>0$, the so-called Clayton copula
\[
C_\theta(u,v)=\bigl(u^{-\theta}+v^{-\theta}-1\bigr)^{-1/\theta},
\]
we obtain a model with the same exponential marginals but positive lower tail
dependence (Section~\ref{subsec:4.8.4}). In this case, large simultaneous values
of $X$ and $Y$ are more likely than under independence, even though the marginal
distributions are unchanged.
\end{enumerate}

This example shows how the copula isolates the dependence structure from the
choice of marginal distributions.
\end{example}

Thus any bivariate distribution with continuous marginals can be decomposed into
its margins and a copula that encodes the dependence structure
\citep{Nelsen1999}. Conversely, by choosing any margins (for example, from the
families in Chapter~3) and any copula, we obtain a valid bivariate distribution.

\subsection{Probability Integral Transform and Invariance}
\label{subsec:4.8.3}

Sklar's theorem is closely tied to the probability integral transform, which
appeared implicitly in Section~2.1 when we discussed distributions of functions
of a random variable \citep[see, e.g.,][]{Nelsen1999}.

If $X$ has continuous cdf $F_X$, then $U=F_X(X)$ is Uniform$(0,1)$; similarly,
if $Y$ has cdf $F_Y$, then $V=F_Y(Y)$ is Uniform$(0,1)$. (Compare this with the
construction of random variables from uniforms in Section~5.5.) If $(X,Y)$ has
copula $C$, then
\[
P(U\le u,V\le v)=C(u,v),\qquad 0<u,v<1 .
\]
Hence copulas are exactly the joint distributions of the transformed pair
$(U,V)=(F_X(X),F_Y(Y))$. This is a specific instance of the
general transformation ideas in Chapter~2.

Copulas are invariant under strictly increasing transformations of the margins.
This invariance is analogous to the invariance principles discussed later in
Chapter~6, but here at the level of probability models rather than estimators.

\begin{theorem}[Invariance]\label{thm:4.8.5}
Let $(X,Y)$ be continuous with copula $C$. Let $g$ and $h$ be strictly
increasing functions, and define $Z=g(X)$, $T=h(Y)$. Then $(Z,T)$ has the same
copula $C$ \citep{Nelsen1999}.
\end{theorem}
\noindent\textbf{Proof:}
Let $F_X$ and $F_Y$ be the marginal cdfs of $X$ and $Y$, and let $F_{X,Y}$ be
their joint cdf. Since $(X,Y)$ has copula $C$, Sklar's Theorem
(Theorem~\ref{thm:4.8.2}) implies
\[
F_{X,Y}(x,y)=C\bigl(F_X(x),F_Y(y)\bigr),\qquad x,y\in\mathbb{R}.
\]
Because $g$ and $h$ are strictly increasing, the events $\{g(X)\le z\}$ and
$\{h(Y)\le t\}$ are equivalent to $\{X\le g^{-1}(z)\}$ and $\{Y\le h^{-1}(t)\}$,
respectively. Thus, for $z,t\in\mathbb{R}$,
\[
P(Z\le z,T\le t)
=P\bigl(X\le g^{-1}(z),\,Y\le h^{-1}(t)\bigr)
=F_{X,Y}\bigl(g^{-1}(z),h^{-1}(t)\bigr).
\]
Let $F_Z$ and $F_T$ be the marginal cdfs of $Z$ and $T$. Since $g$ and $h$ are
strictly increasing,
\[
F_Z(z)=P(Z\le z)=P\bigl(X\le g^{-1}(z)\bigr)=F_X\bigl(g^{-1}(z)\bigr),
\]
\[
F_T(t)=P(T\le t)=P\bigl(Y\le h^{-1}(t)\bigr)=F_Y\bigl(h^{-1}(t)\bigr).
\]
Hence, for $z,t\in\mathbb{R}$,
\begin{eqnarray*}
F_{Z,T}(z,t)
&=& P(Z\le z,T\le t)
=F_{X,Y}\bigl(g^{-1}(z),h^{-1}(t)\bigr)
\\ &=& C\Bigl(F_X\bigl(g^{-1}(z)\bigr),F_Y\bigl(h^{-1}(t)\bigr)\Bigr)
=C\bigl(F_Z(z),F_T(t)\bigr).
\end{eqnarray*}
This is exactly the Sklar representation of $(Z,T)$ with copula $C$, so
$(Z,T)$ has the same copula as $(X,Y)$. \qed

\begin{example}[Ranks and the empirical copula]\label{ex:4.8.2}
Suppose $(X_i,Y_i)$, $i=1,\dots,n$, is a random sample from a continuous
bivariate distribution with copula $C$. Let $R_i$ be the rank of $X_i$ among
$X_1,\dots,X_n$, and let $S_i$ be the rank of $Y_i$ among $Y_1,\dots,Y_n$.
Define the pseudo-observations
\[
U_i=\frac{R_i}{n+1},\qquad V_i=\frac{S_i}{n+1},\qquad i=1,\dots,n.
\]
Each $(U_i,V_i)$ lies in $(0,1)^2$, and these points can be viewed as a sample
from an approximation to $(U,V)=(F_X(X),F_Y(Y))$, whose joint distribution is
the copula $C$.
The empirical copula is then
\[
C_n(u,v)
=\frac{1}{n}\sum_{i=1}^n
\mathbf{1}\{U_i\le u,\,V_i\le v\},\qquad 0\le u,v\le 1.
\]
As $n$ increases, $C_n(u,v)$ converges to $C(u,v)$ for each fixed $(u,v)$, in a
manner analogous to the convergence of the empirical cdf to the true cdf later in
Section~5.4. Thus, the dependence structure can be studied using only the ranks,
without specifying the marginal distributions.
\end{example}

Thus the copula depends only on the rank structure of the data and not on the
particular marginal scales. This is the basis for rank-based inference for
copula models (see Section~7.2, where rank-based procedures are used for
correlation and nonparametric measures of association).

If $g$ or $h$ is decreasing, the copula transforms in a simple way (for example,
via reflections), but the dependence structure remains fully described at the
copula level \citep{Nelsen1999}.

\subsection{Dependence Concepts and Measures}
\label{subsec:4.8.4}

Copulas formalize several notions of dependence and yield measures that are
invariant under monotone transformations \citep{Nelsen1999,Joe1997}. Earlier in
the book, dependence was described in terms of independence (Section~1.3),
covariance, and correlation (Section~2.2). These classical measures depend on
second moments and are not invariant under nonlinear transformations.
Copula-based measures remedy this.

\begin{definition}
    Random variables $X$ and $Y$ are said to be \textbf{positively quadrant
dependent (PQD)} if
\[
P(X\le x,Y\le y)\ge P(X\le x)P(Y\le y)
\]
for all $x,y$. Equivalently, in terms of the copula $C$,
\[
C(u,v)\ge uv,\qquad 0<u,v<1.
\]
\end{definition}
Negative quadrant dependence (NQD) is defined by reversing the inequality.
Because PQD can be expressed entirely in terms of the copula, it is invariant
under strictly increasing transformations of the margins, in contrast with
Pearson correlation (see the discussion at the end of Section~2.2).

\begin{definition}
\textbf{Kendall's $\tau$} is a widely used measure of concordance that depends only on the
copula. For a continuous bivariate distribution with
copula $C$,
\[
\tau
=4\int_0^1\!\!\int_0^1 C(u,v)\,dC(u,v)-1.
\]
Equivalently,
\[
\tau
=P\bigl((X_1-X_2)(Y_1-Y_2)>0\bigr)
-P\bigl((X_1-X_2)(Y_1-Y_2)<0\bigr),
\]
where $(X_1,Y_1)$ and $(X_2,Y_2)$ are independent copies.
\end{definition}
Kendall's $\tau$ satisfies:
\begin{itemize}
\item $-1\le\tau\le 1$, with $\tau=1$ (resp.\ $-1$) if and only if $C=M$ (resp.\
$C=W$), that is, perfect increasing (resp.\ decreasing) functional dependence.
\item $\tau=0$ if $X$ and $Y$ are independent (i.e., $C=\Pi$); compare this with
$\Corr(X,Y)=0$ in Section~2.2.
\item $\tau$ is invariant under strictly monotone transformations of $X$ and/or
$Y$.
\end{itemize}

In practice, $\tau$ is estimated by the sample proportion of concordant minus
discordant pairs, based only on the ranks (see also rank-based methods in
Chapter~7) \citep{Nelsen1999}.

\begin{definition}
\textbf{Spearman's 
$\rho$} is the ordinary correlation of the rank variables and can also be
expressed in terms of the copula. Let $U=F_X(X)$, $V=F_Y(Y)$; then
\[
\rho_S=\Corr(U,V)
=12\int_0^1\!\!\int_0^1\bigl(C(u,v)-uv\bigr)\,du\,dv. 
\]
\end{definition}

\begin{example}[Pearson correlation versus Spearman's $\rho$]\label{ex:4.8.8}
Let $X$ be standard normal, and define $Y=X^3$. Then $Y$ is a strictly
increasing function of $X$, so the copula of $(X,Y)$ is the upper Fr\'echet
bound $M(u,v)=\min\{u,v\}$. In particular, both Kendall's $\tau$ and Spearman's $\rho$
are equal to $1$ for this pair.

However, Pearson's correlation coefficient between $X$ and $Y$ is strictly less
than $1$, because the relationship between $X$ and $Y$ is nonlinear. This
illustrates how copula-based measures such as $\tau$ and $\rho_S$ capture
monotone dependence, while Pearson correlation measures only linear association
(Section~2.2).
\end{example}

Spearman's $\rho$ shares properties analogous to Kendall's $\tau$: it is bounded
between $-1$ and $1$, equals $\pm 1$ for $C=M$ or $C=W$, is $0$ under
independence, and is invariant under strictly increasing transformations of the
margins \citep{Nelsen1999}. Unlike Pearson correlation $\rho$ (Section~2.2),
$\rho_S$ is defined without second-moment assumptions.

These measures are especially useful for parametric copula families, where
$\tau$ or $\rho_S$ can be expressed as simple functions of the dependence
parameter and inverted to give rank-based estimators (cf.\ the method of
moments in Section~7.2).

\begin{definition}
Copulas also well describe extremal dependence.
The \textbf{lower tail
dependence coefficient} is
\[
\lambda_L
=\lim_{u\downarrow 0} P\bigl(Y\le F_Y^{-1}(u)\mid X\le F_X^{-1}(u)\bigr)
=\lim_{u\downarrow 0} \frac{C(u,u)}{u},
\]
when the limit exists. The \textbf{upper tail dependence coefficient} is
\[
\lambda_U
=\lim_{u\uparrow 1} \frac{1-2u+C(u,u)}{1-u}.
\]
\end{definition}

These quantities depend only on the copula and measure the strength of
dependence in the joint lower and upper tails, respectively 
. For example, the Gaussian copula (see next example) with correlation $|\rho|<1$ has
$\lambda_L=\lambda_U=0$, so extreme events are asymptotically independent even
when the linear correlation is high \citep{Joe1997}.

\begin{example}[Gaussian and Gumbel tail behavior]\label{ex:4.8.9}
Consider two dependence models with standard normal marginals:

\begin{enumerate}
\item[(a)] A bivariate normal model with correlation $\rho$, whose copula is
the Gaussian copula $C_\rho$,
\[
C_\rho(u,v)
=\Phi_2\bigl(\Phi^{-1}(u),\Phi^{-1}(v);\rho\bigr),
\]
where $\Phi_2\bigl(\cdot,\cdot;\rho\bigr)$ denotes the bivariate normal cdf with zero means, unit variances, and correlation coefficient $\rho \in (-1,1)$ as in Example~3.3.5.
\item[(b)] A Gumbel copula $C_\theta$ with parameter $\theta\ge 1$, 
\[
C_\theta(u,v)
=\exp\left(
-\bigl((-\log u)^\theta+(-\log v)^\theta\bigr)^{1/\theta}
\right),
\]
combined with
standard normal marginals via Sklar's Theorem.
\end{enumerate}

In both cases, Kendall's $\tau$ and Spearman's $\rho$ can be chosen to be similar by
an appropriate choice of $\theta$, so the overall strength of dependence is
comparable. However, the Gaussian copula has $\lambda_L=\lambda_U=0$, so the
probability of very large simultaneous exceedances decays as if the tails were
asymptotically independent. In contrast, the Gumbel copula has $\lambda_U=2-2^{1/\theta}>0$ (and $\lambda_L=0$), making simultaneous large positive extremes substantially more
likely. 

Concretely, as for the Gaussian copula Kendall's $\tau$ is
\[
\tau_{\text{Gauss}}(\rho)
= \frac{2}{\pi}\arcsin(\rho),
\] setting $\rho=0.8$ yields 
\[
\tau_{\text{Gauss}}(0.8)
= \frac{2}{\pi}\cdot 0.9273
\approx \frac{1.8546}{3.1416}
\approx 0.59.
\]
For the Gumbel copula Kendall's $\tau$ is
$
\tau_{\text{Gumbel}}(\theta)
= 1 - \frac{1}{\theta}.
$
So we set
\[
\tau_{\text{Gumbel}}(\theta)
=\tau_{\text{Gauss}}(0.8)=1 - \frac{1}{\theta}
\approx 0.59.
\]
Solving for $\theta$,
\[
\frac{1}{\theta}\approx 1-0.59=0.41,
\qquad
\theta\approx \frac{1}{0.41}\approx 2.44,
\]
and corresponding $\lambda_U=0.67$. 
\end{example}

This example underscores how copulas allow us to tune not only the
overall dependence but also the extremal (tail) dependence.

\subsection{Visualization of Copulas}
\label{subsec:4.8.5}
Copulas may be visualized most naturally through their behavior on the unit
square. One approach is to plot the copula cdf $C(u,v)$ as a surface over
$(u,v)\in[0,1]^2$; for example, the independence copula appears as the smooth
surface $uv$, while strong positive dependence pulls the surface upward toward
the upper Fr\'echet bound $M(u,v)=\min\{u,v\}$. For absolutely continuous
copulas it is also useful to plot the so-called copula density $c(u,v)
=\frac{\partial^2}{\partial u\,\partial v} C(u,v)$, where regions of
high density indicate values of $(U,V)$ that occur more frequently. In applied
work, empirical copulas are often examined by scatterplots of the
pseudo-observations $(U_i,V_i)$, or by contour plots and heatmaps on the unit
square, which display the dependence structure independently of the marginal
distributions.

\begin{figure}[htbp]
  \centering
  \includegraphics[width=\textwidth]{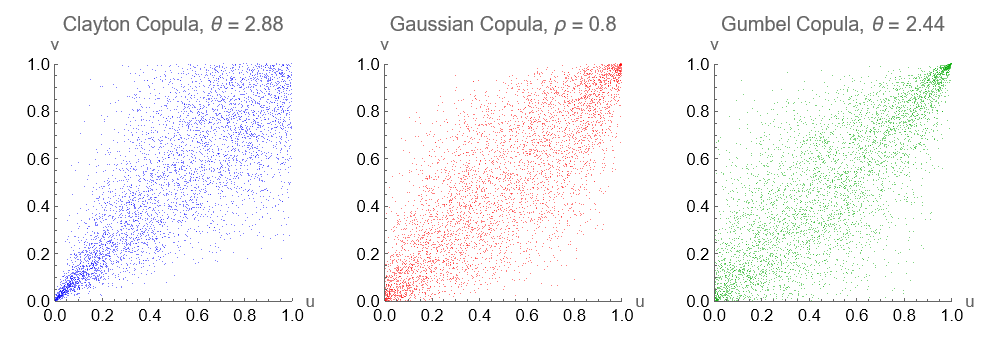}
  \caption{Scatterplots of pseudo-observations from the Clayton copula
    ($\theta=2.88$), Gaussian copula ($\rho=0.8$), and Gumbel copula
    ($\theta=2.44$).}
\label{fig:4.8.1}
\end{figure}

In Figure \ref{fig:4.8.1} we plot three such scatterplots generated by the Mathematica code given in Example A.0.9 in the Appendix. 
These illustrate how different copulas, even with the same margins and similar Kendall's $\tau$, can encode very different dependence structures, particularly in
the tails \citep{Nelsen1999,Joe1997}.

\setcounter{section}{5}
\renewcommand{\thesection}{7.\arabic{section}}
\section{Inference for Copula Models}
\label{sec:7.6}

Suppose that $(X_1,Y_1),\dots,(X_n,Y_n)$ is a random sample from a distribution
with continuous marginals $F_X,F_Y$ and copula $C_\theta$ indexed by a
parameter~$\theta$. The likelihood for $(X,Y)$ decomposes into marginal and
copula parts, paralleling the decomposition of models into components that
underlies the likelihood principle in Chapter~6 \citep[see, e.g.,][]{Joe1997}.

Because the copula is invariant under strictly increasing transformations of the
margins, it is natural to base inference about $\theta$ on the ranks (or,
equivalently, on empirical versions of $U_i=F_X(X_i)$ and $V_i=F_Y(Y_i)$). Let
$R_i$ and $S_i$ be the ranks of $X_i$ and $Y_i$ among $X_1,\dots,X_n$ and
$Y_1,\dots,Y_n$, respectively. The \textbf{empirical copula} is defined on
$[0,1]^2$ by
\[
C_n(u,v)
=\frac{1}{n}
\sum_{i=1}^n
\mathbf{1}\Bigl(\frac{R_i}{n+1}\le u,\frac{S_i}{n+1}\le v\Bigr).
\]

This is a rank-based estimator of $C$; under mild conditions, $C_n(u,v)\to
C(u,v)$ for each fixed $(u,v)$ \citep{GenestFavre2007}. The use of ranks parallels
the rank-based procedures for nonparametric inference introduced in Chapter~8.

While full maximum likelihood estimation is certainly possible, but could be numerically unstable, several other rank-based estimation strategies are used in practice:
\begin{itemize}
\item \emph{Method of moments via Kendall's tau or Spearman's rho.} If
$\tau=\tau(\theta)$ (or $\rho_S=\rho_S(\theta)$) is known in closed form, one
can estimate $\theta$ by solving $\tau(\hat\theta)=\hat\tau$ (or
$\rho_S(\hat\theta)=\hat\rho_S$), where $\hat\tau$ or $\hat\rho_S$ is computed
from the sample ranks (cf.\ Section~7.2) \citep{GenestFavre2007}.
\item \emph{Maximum pseudolikelihood.} If the copula has density
$c_\theta(u,v)$, one may maximize
\[
\ell(\theta)
=\sum_{i=1}^n
\log c_\theta\Bigl(\frac{R_i}{n+1},\frac{S_i}{n+1}\Bigr)
\]
with respect to $\theta$. Under suitable conditions, the maximizer $\hat\theta$
is consistent and asymptotically normal, and its large-sample behavior can be
analyzed using the methods of Chapter~10 \citep{Joe1997}.
\item \emph{Two-step methods.} One may first estimate marginal parameters
by univariate methods (Chapter~7) and then estimate $\theta$ by maximizing a
``copula likelihood'' built from the transformed data, but this approach is
sensitive to misspecification of the margins \citep{Joe1997}.
\end{itemize}

Rank-based methods have the advantage of depending only on the copula and
remaining valid under arbitrary strictly increasing marginal transformations, in
line with the invariance considerations emphasized in Chapter~6
\citep{GenestFavre2007}.

\begin{example}[Estimating a Clayton copula parameter]\label{ex:4.8.5}
Suppose $(X_i,Y_i)$, $i=1,\dots,n$, is a sample from a model with continuous
marginals and Clayton copula
\[
C_\theta(u,v)
=\bigl(u^{-\theta}+v^{-\theta}-1\bigr)^{-1/\theta},\qquad \theta>0.
\]

For this family, Kendall's tau is a simple function of $\theta$:
\[
\tau(\theta)=\frac{\theta}{\theta+2}.
\]
Let $\hat\tau$ be the sample Kendall's tau computed from the ranks of the data.
A method-of-moments estimator of $\theta$ is obtained by solving
$\tau(\hat\theta)=\hat\tau$, which gives
\[
\hat\theta=\frac{2\hat\tau}{1-\hat\tau}.
\]

This estimator depends only on the copula (through the ranks) and not on the
specific marginal distributions, in contrast with likelihood-based estimators
that require a full specification of the marginals (Chapter~7).
\end{example}

\newpage
\subsection*{Exercises for Section 4.10}
\begin{enumerate}

\item[4.66] (Copulas and monotone transformations.) Let $(X,Y)$ be continuous
with copula $C$, and let $Z=g(X)$, $T=h(Y)$ with $g,h$ strictly increasing.
\begin{enumerate}
\item[(a)] Using the change-of-variables ideas from Section~2.1, express the cdf
of $(Z,T)$ in terms of $F_{X,Y}$.
\item[(b)] Show that $(Z,T)$ has the same copula $C$.
\item[(c)] Let $(X,Y)$ be the bivariate normal pair of Example~4.5.3 with
correlation $\rho$. Take $g(x)=e^x$ and $h(y)=e^{3y}$. Simulate a moderate
sample from this model (see Section~5.5) and compare the scatter plots of
$(X,Y)$ and $(Z,T)$. Explain why the copula, and thus Kendall's tau, is the
same in both cases.
\end{enumerate}

\item[4.67] (PQD and covariance sign.) Suppose $(X,Y)$ is continuous and
positively quadrant dependent:
\[
P(X\le x,Y\le y)\ge P(X\le x)P(Y\le y)\quad\text{for all }x,y.
\]
\begin{enumerate}
\item[(a)] Express this condition in terms of the copula $C$.
\item[(b)] Show that if $X$ and $Y$ have finite second moments, then
$\Cov(X,Y)\ge 0$ (hint: proceed as in the proof of Theorem~2.2.6, but use the
PQD condition).
\item[(c)] Consider the model of Example~4.5.2, where $X$ and $Y$ are
independent exponentials with parameter $\lambda$ (hence $\Cov(X,Y)=0$). Modify
this model by introducing a Clayton copula with parameter $\theta>0$ while
keeping the same exponential marginals. Argue that the resulting pair is PQD,
and discuss how the covariance changes.
\end{enumerate}

\item[4.68] (Gaussian copula and tail independence.) Let $(Z_1,Z_2)$ be
bivariate normal with correlation $\rho$ (Section~3.3), and define
$U=\Phi(Z_1)$, $V=\Phi(Z_2)$.
\begin{enumerate}
\item[(a)] Argue that the copula of $(U,V)$ is the Gaussian copula $C_\rho$.
\item[(b)] Using the tail dependence formulas in
Section~\ref{subsec:4.8.4}, show heuristically that for $|\rho|<1$,
$\lambda_L=\lambda_U=0$.
\item[(c)] Explain in words why this means that even strongly correlated Gaussian
models may understate the probability of joint extreme events.
\end{enumerate}

\end{enumerate}

\newpage
\subsection*{Exercises for Section 7.6}
\begin{enumerate}
\item[7.67] (Kendall's $\tau$ and Archimedean copulas.) An Archimedean copula has
the form
\[
C(u,v)=\varphi^{-1}(\varphi(u)+\varphi(v)),
\]
where $\varphi$ is a convex, decreasing generator.

\begin{enumerate}
\item[(a)] For the Gumbel copula with parameter $\theta>0$,
verify that the generator is $\varphi(t)=t^{-\theta}-1$.
\item[(b)] Using the following formula for Kendall's $\tau$ in terms of the generator 
\[
\tau = 1 + 4 \int_0^1 \frac{\varphi(t)}{\varphi'(t)}\,dt,
\] derive $\tau(\theta)$ and base an estimator for $\theta$ on it.
\end{enumerate}

\end{enumerate}

\subsection*{Further Reading}

For a systematic treatment of copulas, their properties, and many families and
examples, see \citet{Nelsen1999}. A more advanced monograph with emphasis on
multivariate dependence and extreme-value structures is \citet{Joe1997}. For an
accessible introduction to inference for copula models based on ranks, with a
detailed worked example and a hydrological application, see
\citet{GenestFavre2007}.

\bibliographystyle{plainnat} 
\bibliography{copulas}

\appendix

\section*{Appendix}

\addcontentsline{toc}{section}{Appendix A: Mathematica Code for Copula Scatterplots}

\textbf{Example A.0.9 (Copula Scatterplots)}
 The following Mathematica code generates scatterplots of pseudo-observations
from three copula models with comparable overall dependence: a Clayton copula
with parameter $\theta=2.88$, a Gaussian copula with correlation $\rho=0.8$,
and a Gumbel copula with parameter $\theta=2.44$.

\begin{verbatim}
(* ============================================== *)
(* Copula scatterplots: Clayton, Gaussian, Gumbel *)
(* ============================================== *)

SeedRandom[12345];
n = 5000;

(* ---------------------------------------------- *)
(* 1. Clayton copula, parameter theta = 2.88      *)
(* ---------------------------------------------- *)

thetaClayton = 2.88;

(* Correct Clayton sampler using conditional distribution *)
claytonSample[n_, theta_] := Module[{u, w, v},
  u = RandomReal[{0, 1}, n];
  w = RandomReal[{0, 1}, n];
  v = ((w^(1/(1 + theta)) * u)^(-theta) + 1 - u^(-theta))^(-1/theta);
  Transpose[{u, v}]
];

claytonData = claytonSample[n, thetaClayton];

(* ---------------------------------------------- *)
(* 2. Gaussian copula, parameter rho = 0.8        *)
(* ---------------------------------------------- *)

rho = 0.8;
sigma = {{1, rho}, {rho, 1}};

gaussianSample[n_, sigma_] := Module[{z},
  z = RandomVariate[MultinormalDistribution[{0, 0}, sigma], n];
  CDF[NormalDistribution[0, 1], #] & /@ z
];

gaussianData = gaussianSample[n, sigma];

(* ---------------------------------------------- *)
(* 3. Gumbel copula, parameter theta = 2.44       *)
(* ---------------------------------------------- *)

thetaGumbel = 2.44;

(* Positive alpha-stable sampler, alpha = 1/theta *)
positiveStableSample[alpha_] := Module[{V, E},
  V = RandomReal[{0, Pi}, 1][[1]];
  E = RandomVariate[ExponentialDistribution[1]];
  (Sin[alpha V]/(Sin[V])^(1/alpha)) *
   (Sin[(1 - alpha) V]/E)^((1 - alpha)/alpha)
];

(* Marshall-Olkin sampler for Gumbel copula *)
gumbelSample[n_, theta_] := Module[{alpha, s, e1, e2, u, v},
  alpha = 1/theta;
  Table[
    s = positiveStableSample[alpha];
    e1 = RandomVariate[ExponentialDistribution[1]];
    e2 = RandomVariate[ExponentialDistribution[1]];
    u = Exp[-(e1/s)^alpha];
    v = Exp[-(e2/s)^alpha];
    {u, v},
    {n}
  ]
];

gumbelData = gumbelSample[n, thetaGumbel];

(* ---------------------------------------------- *)
(* 4. Scatterplots                                *)
(* ---------------------------------------------- *)

plotStyle[color_] := Directive[color, PointSize[0.0025], Opacity[0.5]];

claytonPlot = ListPlot[
  claytonData,
  PlotRange -> {{0, 1}, {0, 1}},
  AspectRatio -> 1,
  PlotStyle -> plotStyle[Blue],
  AxesLabel -> {"u", "v"},
  PlotLabel -> "Clayton Copula, \[Theta] = 2.88",
  ImageSize -> 350
];

gaussianPlot = ListPlot[
  gaussianData,
  PlotRange -> {{0, 1}, {0, 1}},
  AspectRatio -> 1,
  PlotStyle -> plotStyle[Red],
  AxesLabel -> {"u", "v"},
  PlotLabel -> "Gaussian Copula, \[Rho] = 0.8",
  ImageSize -> 350
];

gumbelPlot = ListPlot[
  gumbelData,
  PlotRange -> {{0, 1}, {0, 1}},
  AspectRatio -> 1,
  PlotStyle -> plotStyle[Darker[Green]],
  AxesLabel -> {"u", "v"},
  PlotLabel -> "Gumbel Copula, \[Theta] = 2.44",
  ImageSize -> 350
];

GraphicsRow[{claytonPlot, gaussianPlot, gumbelPlot}, Spacings -> 20]
\end{verbatim}

\newpage
\section*{Solutions}

\textbf{Solution to 4.66.}

Let $(X,Y)$ be a continuous bivariate random vector with joint cdf $F_{X,Y}$ and copula $C$, and let $g,h:\mathbb{R}\to\mathbb{R}$ be strictly increasing functions.  Define
\[
Z = g(X),\qquad T = h(Y).
\]

\bigskip

\noindent
\textbf{(a) Cdf of $(Z,T)$ in terms of $F_{X,Y}$.}

For any $z,t\in\mathbb{R}$,
\[
\begin{aligned}
F_{Z,T}(z,t)
&= \mathbb{P}(Z \le z,\; T \le t) \\
&= \mathbb{P}\bigl(g(X) \le z,\; h(Y) \le t\bigr).
\end{aligned}
\]
Since $g$ and $h$ are strictly increasing, the inequalities $g(X)\le z$ and $h(Y)\le t$ are equivalent to
\[
X \le g^{-1}(z),\qquad Y \le h^{-1}(t),
\]
where $g^{-1}$ and $h^{-1}$ denote the (strictly increasing) inverses of $g$ and $h$.  Hence
\[
F_{Z,T}(z,t)
= \mathbb{P}\bigl(X \le g^{-1}(z),\; Y \le h^{-1}(t)\bigr)
= F_{X,Y}\bigl(g^{-1}(z),\,h^{-1}(t)\bigr).
\]

\bigskip

\noindent
\textbf{(b) Copula of $(Z,T)$.}

Let $F_X$ and $F_Y$ be the marginal cdfs of $X$ and $Y$, and let $F_Z$ and $F_T$ be the marginal cdfs of $Z$ and $T$.  From part (a),
\[
\begin{aligned}
F_Z(z)
&= \mathbb{P}(Z \le z)
= \mathbb{P}\bigl(g(X) \le z\bigr)
= \mathbb{P}\bigl(X \le g^{-1}(z)\bigr)
= F_X\bigl(g^{-1}(z)\bigr), \\
F_T(t)
&= \mathbb{P}(T \le t)
= \mathbb{P}\bigl(h(Y) \le t\bigr)
= \mathbb{P}\bigl(Y \le h^{-1}(t)\bigr)
= F_Y\bigl(h^{-1}(t)\bigr).
\end{aligned}
\]

Let $u,v\in(0,1)$ and define
\[
z = F_Z^{-1}(u),\qquad t = F_T^{-1}(v).
\]
Using $F_Z(z)=F_X(g^{-1}(z))$ and the strict monotonicity of $g$,
\[
F_X\bigl(g^{-1}(z)\bigr) = u
\quad\Longrightarrow\quad
g^{-1}(z) = F_X^{-1}(u),
\]
so that
\[
z = g\bigl(F_X^{-1}(u)\bigr).
\]
Similarly,
\[
t = h\bigl(F_Y^{-1}(v)\bigr).
\]

The copula $C_{Z,T}$ of $(Z,T)$ is defined by
\[
C_{Z,T}(u,v)
= F_{Z,T}\bigl(F_Z^{-1}(u),\,F_T^{-1}(v)\bigr).
\]
Using part (a) and the relations above,
\[
\begin{aligned}
C_{Z,T}(u,v)
&= F_{Z,T}(z,t) \\
&= F_{X,Y}\bigl(g^{-1}(z),\,h^{-1}(t)\bigr) \\
&= F_{X,Y}\bigl(F_X^{-1}(u),\,F_Y^{-1}(v)\bigr).
\end{aligned}
\]
By definition of the copula $C$ of $(X,Y)$,
\[
C(u,v)
= F_{X,Y}\bigl(F_X^{-1}(u),\,F_Y^{-1}(v)\bigr),
\]
hence
\[
C_{Z,T}(u,v) = C(u,v),\qquad (u,v)\in[0,1]^2.
\]
Therefore $(Z,T)$ has the same copula $C$ as $(X,Y)$.

\bigskip

\noindent
\textbf{(c) Bivariate normal example and invariance of the copula.}

Let $(X,Y)$ be bivariate normal with correlation $\rho$, as in Example~4.5.3, and define
\[
Z = g(X) = e^{X},\qquad T = h(Y) = e^{3Y}.
\]
To simulate a sample of size $n$ from this model, one can:

\begin{enumerate}
\item Generate $n$ i.i.d.\ observations $(X_i,Y_i)$ from the bivariate normal distribution with correlation $\rho$.
\item For each $i$, set
\[
Z_i = e^{X_i},\qquad T_i = e^{3Y_i}.
\]
\end{enumerate}

The scatter plot of $(Z_i,T_i)$ is obtained from that of $(X_i,Y_i)$ by applying strictly increasing transformations in each coordinate.  These change the marginal distributions (normal to lognormal and scaled lognormal) and the scale of the axes, but they preserve the ordering of the points in each coordinate and hence the dependence structure encoded by the copula.

Part (b) shows rigorously that strictly increasing transformations in each margin leave the copula unchanged.  Kendall's tau is a measure of concordance that depends only on the copula (it is invariant under strictly increasing transformations of each coordinate).  Therefore $(X,Y)$ and $(Z,T)$ have the same copula and the same value of Kendall's tau; empirically, the estimated tau from the two scatter plots will (up to sampling error) be equal.

\bigskip
\noindent
\textbf{Solution to 4.67.}

Assume $(X,Y)$ is continuous with joint cdf $F_{X,Y}$, marginals $F_X,F_Y$, and copula $C$.

\bigskip

\noindent
\textbf{(a) PQD in terms of the copula.}

Positive quadrant dependence (PQD) means
\[
P(X\le x,Y\le y)\;\ge\;P(X\le x)\,P(Y\le y)\quad\text{for all }x,y\in\mathbb{R}.
\]
In terms of the cdfs this is
\[
F_{X,Y}(x,y)\;\ge\;F_X(x)\,F_Y(y)\quad\forall x,y.
\]
By Sklar’s theorem,
\[
F_{X,Y}(x,y)=C\bigl(F_X(x),F_Y(y)\bigr).
\]
Let $u=F_X(x)$ and $v=F_Y(y)$; then $u,v\in[0,1]$ and the PQD condition becomes
\[
C(u,v)\;\ge\;uv\quad\text{for all }(u,v)\in[0,1]^2.
\]
Thus $(X,Y)$ is PQD if and only if its copula $C$ satisfies $C(u,v)\ge uv$ on $[0,1]^2$.

\bigskip

\noindent
\textbf{(b) PQD implies nonnegative covariance.}

Assume $X$ and $Y$ have finite second moments.  We show $\Cov(X,Y)\ge0$.
Recall that for square–integrable $(X,Y)$ one can write
\[
\Cov(X,Y)
= \mathbb{E}(XY)-\mathbb{E}X\,\mathbb{E}Y
  = \iint_{\mathbb{R}^2}
     \bigl\{P(X>x,Y>y)-P(X>x)\,P(Y>y)\bigr\}\,dx\,dy,
\]
which is obtained by expanding $XY$ into indicator integrals and interchanging expectation and integration (as in the proof of Theorem~2.2.6).

Note that
\[
P(X>x,Y>y)
= 1-F_X(x)-F_Y(y)+F_{X,Y}(x,y),
\]
and
\[
P(X>x)\,P(Y>y)
= (1-F_X(x))(1-F_Y(y))
= 1-F_X(x)-F_Y(y)+F_X(x)F_Y(y).
\]
Hence
\[
P(X>x,Y>y)-P(X>x)\,P(Y>y)
= F_{X,Y}(x,y)-F_X(x)F_Y(y).
\]
Therefore
\[
\Cov(X,Y)
= \iint_{\mathbb{R}^2}
   \bigl(F_{X,Y}(x,y)-F_X(x)F_Y(y)\bigr)\,dx\,dy.
\]

If $(X,Y)$ is PQD, then $F_{X,Y}(x,y)\ge F_X(x)F_Y(y)$ for all $x,y$, so the integrand is everywhere nonnegative.  The integral of a nonnegative function is nonnegative, hence
\[
\Cov(X,Y)\;\ge\;0.
\]

\bigskip

\noindent
\textbf{(c) Clayton copula with exponential marginals.}

In Example~4.5.2, $X$ and $Y$ are independent exponentials with parameter $\lambda$:
\[
F_X(x)=F_Y(x)=1-e^{-\lambda x},\qquad x\ge0,
\]
and
\[
F_{X,Y}(x,y)=F_X(x)F_Y(y)
\quad\Rightarrow\quad
\Cov(X,Y)=0.
\]

Now keep these marginals but introduce dependence via a Clayton copula with parameter $\theta>0$:
\[
C_\theta(u,v)
= \left(u^{-\theta}+v^{-\theta}-1\right)^{-1/\theta},
\qquad u,v\in(0,1],\;\theta>0.
\]
The joint cdf of the new pair $(X_\theta,Y_\theta)$ is
\[
F_{X_\theta,Y_\theta}(x,y)
= C_\theta\bigl(F_X(x),F_Y(y)\bigr),\qquad x,y\ge0.
\]

\emph{PQD:}  For Clayton copulas with $\theta>0$ one has
\[
C_\theta(u,v)\;\ge\;uv,\qquad (u,v)\in[0,1]^2.
\]
Thus the copula $C_\theta$ lies above the independence copula $uv$, so by part (a) the resulting pair $(X_\theta,Y_\theta)$ is positively quadrant dependent.

\emph{Covariance:}  By part (b), PQD and finite second moments imply
\[
\Cov(X_\theta,Y_\theta)\;\ge\;0.
\]
When $\theta=0$ (the limit case) the copula reduces to the product copula $uv$, i.e.\ independence, so $\Cov(X_0,Y_0)=0$.  For $\theta>0$, the dependence becomes increasingly positive (stronger association of large values with large values and of small with small), and $\Cov(X_\theta,Y_\theta)$ becomes strictly positive and increases with $\theta$.  Thus introducing a Clayton copula with parameter $\theta>0$ while keeping exponential marginals turns the originally independent pair into a PQD pair with positive covariance, whose covariance grows as $\theta$ increases.

\bigskip
\noindent
\textbf{Solution to 4.68.}

Let $(Z_1,Z_2)$ be bivariate normal with correlation $\rho$, and let
\[
U=\Phi(Z_1),\qquad V=\Phi(Z_2),
\]
where $\Phi$ is the standard normal cdf.

\bigskip

\noindent
\textbf{(a) Copula of $(U,V)$.}

By construction,
\[
U = F_{Z_1}(Z_1),\qquad V = F_{Z_2}(Z_2),
\]
where $F_{Z_i}=\Phi$ are the marginal cdfs of $Z_i$.  Thus $U$ and $V$ are standard uniform random variables (probability integral transform).  Their joint cdf is
\[
\begin{aligned}
P(U\le u, V\le v)
&= P\bigl(\Phi(Z_1)\le u,\ \Phi(Z_2)\le v\bigr) \\
&= P\bigl(Z_1\le \Phi^{-1}(u),\ Z_2\le \Phi^{-1}(v)\bigr) \\
&= F_{(Z_1,Z_2)}\bigl(\Phi^{-1}(u),\Phi^{-1}(v)\bigr),
\end{aligned}
\]
where $F_{(Z_1,Z_2)}$ is the bivariate normal cdf with correlation $\rho$.
By definition, the copula $C_\rho$ associated with this bivariate normal distribution is
\[
C_\rho(u,v)
:= F_{(Z_1,Z_2)}\bigl(\Phi^{-1}(u),\Phi^{-1}(v)\bigr).
\]
Comparing the two displays, we see that $(U,V)$ has copula $C_\rho$ (the Gaussian copula with parameter $\rho$).

\bigskip

\noindent
\textbf{(b) Tail dependence for $|\rho|<1$.}

Recall the upper and lower tail dependence coefficients for a copula $C$:
\[
\lambda_U
= \lim_{u\uparrow 1} \frac{1 - 2u + C(u,u)}{1-u},
\qquad
\lambda_L
= \lim_{u\downarrow 0} \frac{C(u,u)}{u}.
\]

For the Gaussian copula $C_\rho$, we have
\[
C_\rho(u,u)
= P(U\le u, V\le u)
= P\bigl(Z_1\le z,\ Z_2\le z\bigr),
\quad\text{where }z=\Phi^{-1}(u).
\]
Equivalently,
\[
P(Z_1>z, Z_2>z)
= 1 + C_\rho(u,u) - 2u.
\]

\emph{Upper tail.}
Write
\[
P(U>u, V>u)
= P(Z_1>z, Z_2>z),
\quad u\to 1,\;z\to\infty.
\]
For a bivariate normal with $|\rho|<1$, it is known (and can be derived via asymptotic analysis of the joint normal tail) that
\[
P(Z_1>z, Z_2>z)
\sim k(\rho)\,\frac{1}{z}\exp\!\left(-\frac{z^2}{1+\rho}\right),
\qquad z\to\infty,
\]
for some positive constant $k(\rho)$ depending on $\rho$.  On the other hand,
\[
P(Z_1>z)
= 1-u
\sim \frac{1}{z}\exp\!\left(-\frac{z^2}{2}\right),
\qquad z\to\infty.
\]
Therefore
\[
\frac{P(Z_1>z, Z_2>z)}{P(Z_1>z)}
\sim k(\rho)\exp\!\left(-z^2\left[\frac{1}{1+\rho} - \frac{1}{2}\right]\right).
\]
If $|\rho|<1$, then $\frac{1}{1+\rho} > \frac{1}{2}$, so the exponent is negative and the ratio decays to $0$ as $z\to\infty$.  Translating back to the copula formulation, this means
\[
\lambda_U
= \lim_{u\uparrow 1} P(U>u\mid V>u)
= 0,\qquad |\rho|<1.
\]

\emph{Lower tail.}
By symmetry of the bivariate normal about $(0,0)$, $(Z_1,Z_2)$ and $(-Z_1,-Z_2)$ have the same joint distribution (with the same correlation $\rho$).  Thus lower–tail events $(Z_1\le -z, Z_2\le -z)$ correspond to upper–tail events for $(-Z_1,-Z_2)$; the same asymptotic analysis applies and yields
\[
\lambda_L = 0,\qquad |\rho|<1.
\]

Heuristically, for any fixed $\rho$ strictly less than $1$ in magnitude, the joint tail of the bivariate normal decays faster than the marginal tail, so the conditional probability of a joint extreme given a marginal extreme goes to zero in both tails.

\bigskip

\noindent
\textbf{(c) Interpretation for joint extremes.}

The fact that $\lambda_L=\lambda_U=0$ for $|\rho|<1$ means that, in a Gaussian copula model, the probability of observing one variable in an extreme tail given that the other is in that tail tends to $0$ as we move further into the tail.  In other words, even if the linear correlation $\rho$ is large (say $\rho=0.9$), the Gaussian dependence structure does not create asymptotic clustering of extremes: very large values of $X$ and $Y$ do \emph{not} tend to occur together with positive limiting probability.

In practical terms, this implies that strongly correlated Gaussian models can substantially understate the probability of joint extreme events.  They may fit well in the center of the distribution (capturing linear correlation), but they impose tail independence: extreme events in one margin are, asymptotically, almost never accompanied by simultaneous extremes in the other margin.  For applications where joint tail behavior is critical (e.g.\ financial or environmental risk), the Gaussian copula may therefore give overly optimistic (too small) estimates of joint extreme risks.

\bigskip
\noindent
\textbf{Solution to 7.67}

An Archimedean copula has the form
\[
C(u,v) = \varphi^{-1}\bigl(\varphi(u)+\varphi(v)\bigr),
\]
where $\varphi$ is a convex, decreasing generator.

\begin{enumerate}
\item[(a)] \emph{Gumbel generator.}

The (bivariate) Gumbel copula with parameter $\theta>0$ is
\[
C_\theta(u,v)
= \exp\!\left(
   -\bigl((-\ln u)^\theta + (-\ln v)^\theta\bigr)^{1/\theta}
  \right), 
  \qquad u,v\in(0,1).
\]

We claim that this can be written in Archimedean form with generator
\[
\varphi(t) = t^{-\theta} - 1,\qquad t\in(0,1].
\]

First compute its inverse:
\[
s = \varphi(t) = t^{-\theta} - 1
\quad\Longrightarrow\quad
t^{-\theta} = s+1
\quad\Longrightarrow\quad
t = (1+s)^{-1/\theta}.
\]
Hence
\[
\varphi^{-1}(s) = (1+s)^{-1/\theta},\qquad s\ge0.
\]

Now form the Archimedean copula associated with this generator:
\[
C(u,v)
= \varphi^{-1}\bigl(\varphi(u)+\varphi(v)\bigr)
= \Bigl(
    1 + (u^{-\theta}-1) + (v^{-\theta}-1)
  \Bigr)^{-1/\theta}
= \bigl(u^{-\theta} + v^{-\theta} - 1\bigr)^{-1/\theta}.
\]

This is exactly the (alternative) Archimedean representation of the
Gumbel copula, so $\varphi(t)=t^{-\theta}-1$ is indeed a valid generator
for the Gumbel family.

\item[(b)] \emph{Kendall's $\tau$ and an estimator for $\theta$.}

For an Archimedean copula with generator $\varphi$, Kendall's $\tau$ is
given by
\[
\tau = 1 + 4\int_0^1 \frac{\varphi(t)}{\varphi'(t)}\,dt.
\]

For $\varphi(t)=t^{-\theta}-1$ we have
\[
\varphi'(t) = -\theta t^{-\theta-1}.
\]
Therefore
\[
\frac{\varphi(t)}{\varphi'(t)}
= \frac{t^{-\theta}-1}{-\theta t^{-\theta-1}}
= -\frac{1}{\theta}\,(t - t^{\theta+1}).
\]

Plugging into the formula for $\tau$,
\[
\tau(\theta)
= 1 + 4\int_0^1 \frac{\varphi(t)}{\varphi'(t)}\,dt
= 1 + 4\int_0^1 \left(-\frac{1}{\theta}(t - t^{\theta+1})\right)\!dt
= 1 - \frac{4}{\theta}\int_0^1 (t - t^{\theta+1})\,dt.
\]

Compute the integral:
\[
\int_0^1 t\,dt = \frac{1}{2},
\qquad
\int_0^1 t^{\theta+1}\,dt = \frac{1}{\theta+2}.
\]
Hence
\[
\int_0^1 (t - t^{\theta+1})\,dt
= \frac{1}{2} - \frac{1}{\theta+2}
= \frac{\theta}{2(\theta+2)}.
\]

Substitute back:
\[
\tau(\theta)
= 1 - \frac{4}{\theta}\cdot\frac{\theta}{2(\theta+2)}
= 1 - \frac{2}{\theta+2}
= \frac{\theta}{\theta+2}.
\]

Thus Kendall's $\tau$ for this Gumbel copula (with generator
$\varphi(t)=t^{-\theta}-1$) is
\[
\boxed{\tau(\theta) = \frac{\theta}{\theta+2}}.
\]

To base an estimator for $\theta$ on Kendall's $\tau$, let $\hat\tau$
denote the sample (empirical) Kendall's $\tau$. We invert the above
relationship:
\[
\hat\tau = \frac{\theta}{\theta+2}
\quad\Longrightarrow\quad
\hat\tau(\theta+2) = \theta
\quad\Longrightarrow\quad
\theta(1 - \hat\tau) = 2\hat\tau
\quad\Longrightarrow\quad
\boxed{\hat\theta = \frac{2\hat\tau}{1-\hat\tau}}.
\]

This $\hat\theta$ is the method-of-moments–type estimator of the Gumbel
parameter based on Kendall's $\tau$.
\end{enumerate}
\end{document}